\documentclass[conference]{IEEEtran}
\usepackage{cite}
\usepackage{amsmath,amssymb,amsfonts}
\usepackage{algorithmic}
\usepackage{graphicx}
\usepackage{textcomp}
\usepackage{amsmath}
\usepackage{bm}
\usepackage{amsmath,bbm,amssymb,amsfonts,amstext,amsopn}
\usepackage{amsthm}
\usepackage{amssymb}
\usepackage{epstopdf}

\usepackage{amsthm}
\usepackage{cite}
\usepackage{balance}
\allowbreak
\usepackage{url}
\usepackage{epsfig}
\usepackage{setspace}
\usepackage{stmaryrd}
\usepackage{psfrag}	
\usepackage{multirow}
\usepackage{float}
\usepackage{amsthm}

\newcommand{\subparagraph}{}
\usepackage[compact]{titlesec}
\usepackage[english]{babel}
\usepackage[process=auto]{pstool}
\usepackage{amsthm}
\theoremstyle{remark}
\newtheorem{remark}{Remark}
\usepackage{etoolbox}
\usepackage{algorithm}
\usepackage[process=auto]{pstool}
\usepackage{epsfig}
\usepackage{caption}
\usepackage{algorithmic}
\usepackage{xcolor}
\newtheorem{thm}{Theorem}
\newtheorem{lem}{Lemma}

\begin{document}
	\captionsetup[figure]{labelformat={default},labelsep=period,name={Fig.}}
\title{Resource Allocation for Multi-User Downlink URLLC-OFDMA Systems \vspace*{-4mm}}
\author{ Walid R. Ghanem,
Vahid Jamali, Yan Sun, and Robert Schober  \\
Friedrich-Alexander-University Erlangen-Nuremberg, Germany \quad}
\maketitle \vspace*{-5mm}
\begin{abstract}
In this paper, we study resource allocation for downlink orthogonal frequency division multiple access (OFDMA) systems with the objective to enable ultra-reliable low latency communication (URLLC). To meet the stringent delay requirements of URLLC, the impact of short packet communication is taken into account for resource allocation algorithm design. In particular, the resource allocation design is formulated as a weighted system sum throughput maximization problem with quality-of-service (QoS) constraints for the URLLC users. The formulated problem is non-convex, and hence,  finding the global optimum solution entails a high computational complexity. Thus, an algorithm based on successive convex optimization is proposed to find a sub-optimal solution with polynomial time complexity. Simulation results confirm that the proposed algorithm facilities URLLC and significantly improves the system sum throughput compared to two benchmark schemes. 
\end{abstract}
\renewcommand{\baselinestretch}{0.7}
 \section{Introduction} 
 Ultra-reliable low latency communication (URLLC) is an essential component of the fifth-generation (5G) wireless networks. URLLC is required for mission critical applications, such as autonomous driving, e-health, factory automation, and tactile internet \cite{Toward}. URLLC imposes strict quality-of-service (QoS) requirements including a very low latency (e.g., $1\,$ms) and a low packet error probability (e.g., $10^{-6})$ \cite{Toward}. In addition, the data packet size is typically small, e.g., around 160~bits \cite{Popovski1}. Unfortunately, existing mobile communication systems cannot meet these requirements. For example, for the long term evolution system (LTE), the total frame time is 10$\,$ms, which already exceeds the total latency requirement of URLLC \cite{Mehdi1}. The main challenges for the design of URLLC systems are the two contradicting requirements of low latency and ultra high reliability. For this reason, new design strategies are needed to enable URLLC.
 
  The design of URLLC systems requires a short frame structure and a small packet size. For small packet sizes, the relation between the achievable rate, decoding error probability, and transmission delay cannot be captured by Shannon's capacity formula which assumes infinite block length and zero error probability \cite{shannon}. If Shannon's capacity formula is utilized for resource allocation design for URLLC, the latency will be underestimated and the reliability will be overestimated, and the QoS requirements cannot be met. To address this issue, the relevant performance limits for short packet communication (SPC) have to be considered \cite{thesis}. These performance limits provide a relation between the achievable rate, decoding error probability, and packet length. The pioneering work in \cite{strassen} investigated the limits of SPC for discrete memoryless channels, while the authors in \cite{Polyanskiy} extended this analysis to different types of channels, including the additive white Gaussian noise (AWGN) channel and the Gilbert-Elliot channel. SPC for parallel Gaussian channels was analyzed in \cite{thesis}, while in \cite{Erseghe1} an asymptotic analysis based on the Laplace integral was provided for the AWGN channel, parallel AWGN channels, and the binary symmetric channel (BSC). In \cite{Quasi}, the authors investigated the maximum achievable rate for SPC over quasi-static multiple-input multiple-output fading channels. 
  
     Modern communication systems employ multi-carrier transmission, e.g., orthogonal frequency division multiple access (OFDMA), due to its ability to exploit multi-user diversity and the flexibility in the allocation of resources, such as, e.g., power and bandwidth. Hence, future communication systems will combine the concepts of OFDMA and URLLC. To achieve high performance, the resource allocation in such URLLC-OFDMA systems has to be optimized. However, existing resource allocation designs for URLLC systems assume single-carrier transmission  \cite{optimal,convexfinite,chsecross}. In particular, the authors in \cite{optimal} investigated optimal power allocation for maximizing the average throughput in multi-user time division multiple access (TDMA) networks. The authors in \cite{convexfinite} studied energy efficient packet scheduling over quasi static fading channels. In  \cite{chsecross},
     a cross-layer framework based on the effective
     bandwidth was proposed for optimal resource allocation under QoS constraints. However, single-carrier systems suffer from a poor spectrum utilization, and require complex equalization at the receiver. On the other hand, existing resource allocation algorithms for OFDMA systems, such as those in \cite{wong,kwan1}, are based on Shannon's capacity theorem which assumes infinite block length. Therefore, for URLLC-OFDMA systems, the obtained resource allocation policies are invalid. To the best of the authors' knowledge, the resource allocation design for broadband OFDMA systems providing URLLC has not been investigated in the literature, yet. Motivated by the above discussion, this paper provides the following main contributions:  
   \begin{itemize}
   \item We propose a novel resource allocation algorithm design for multi-user URLLC-OFDMA systems. The resource allocation algorithm design is formulated as an optimization problem with the objective to maximize the weighted sum throughput subject to QoS constraints for all URLLC users. The QoS constraints include the minimum number of transmitted bits, the maximum packet error probability, and the maximum time for transmission of a packet (i.e., the maximum delay).
   \item The formulated problem is a non-convex mixed integer nonlinear optimization problem, and finding the global optimum solution requires high computational complexity. Therefore, a low-complexity sub-optimal algorithm is developed, which obtains a local optimal solution based on successive convex approximation.
   \item We show by simulation that the proposed scheme facilities URLLC, and achieves significant performance gains compared to two benchmark schemes.  
\end{itemize}
\textit{Notation}: In this paper, lower-case letters $x$ refer to scalar numbers, while bold lower-case letters $\mathbf{x}$ represent vectors. $\log(\cdot)$ is the logarithm with base 2. $(\cdot)^{T}$ denotes the transpose operator. 
$\mathbb{R}^{N \times 1}$ represents the set of all $N \times 1$ vectors with real valued entries. The circularly symmetric complex Gaussian distribution with mean $\mu$ and variance $\sigma^{2}$ is denoted by $\mathcal{CN}(\mu,\sigma^{2})$, and $\sim$ stands for "distributed as", $\mathcal{E}\{\cdot\}$ denotes statistical expectation. $\nabla_{x}f(\mathbf{x})$ denotes the gradient vector of function $f(\mathbf{x})$ and its elements are the partial derivatives of $f(\mathbf{x})$.
\break
\vspace*{-4mm}

\section{System and Channel Models}
In this section, we present the considered system and channel models.
\subsection{System Model}
We consider a single-cell downlink OFDMA system, where a base station (BS) serves $K$ URLLC users indexed by $k =\{1,\dots,K\}$, cf. Fig.~ \ref{model}(a). The entire frequency  band is divided into $M$ orthogonal subcarriers indexed by $m \in \{1,\dots,M\}$. We assume that a resource frame has a duration of $T_{\mathrm{f}}$ seconds, and consists of $N$ time slots which are indexed by $n \in \{1,\dots,N\}$. Thereby, one OFDMA symbol spans one time slot. Therefore, we have in total $M \times N$ resource elements. We assume that the delay requirement of each user is known at the BS and only users whose delay requirements can potentially be met in the current resource block are admitted into the system. The maximum transmit power of the BS is $P_{\text{max}}$.
 \begin{figure}[t]
\includegraphics[scale=0.31]{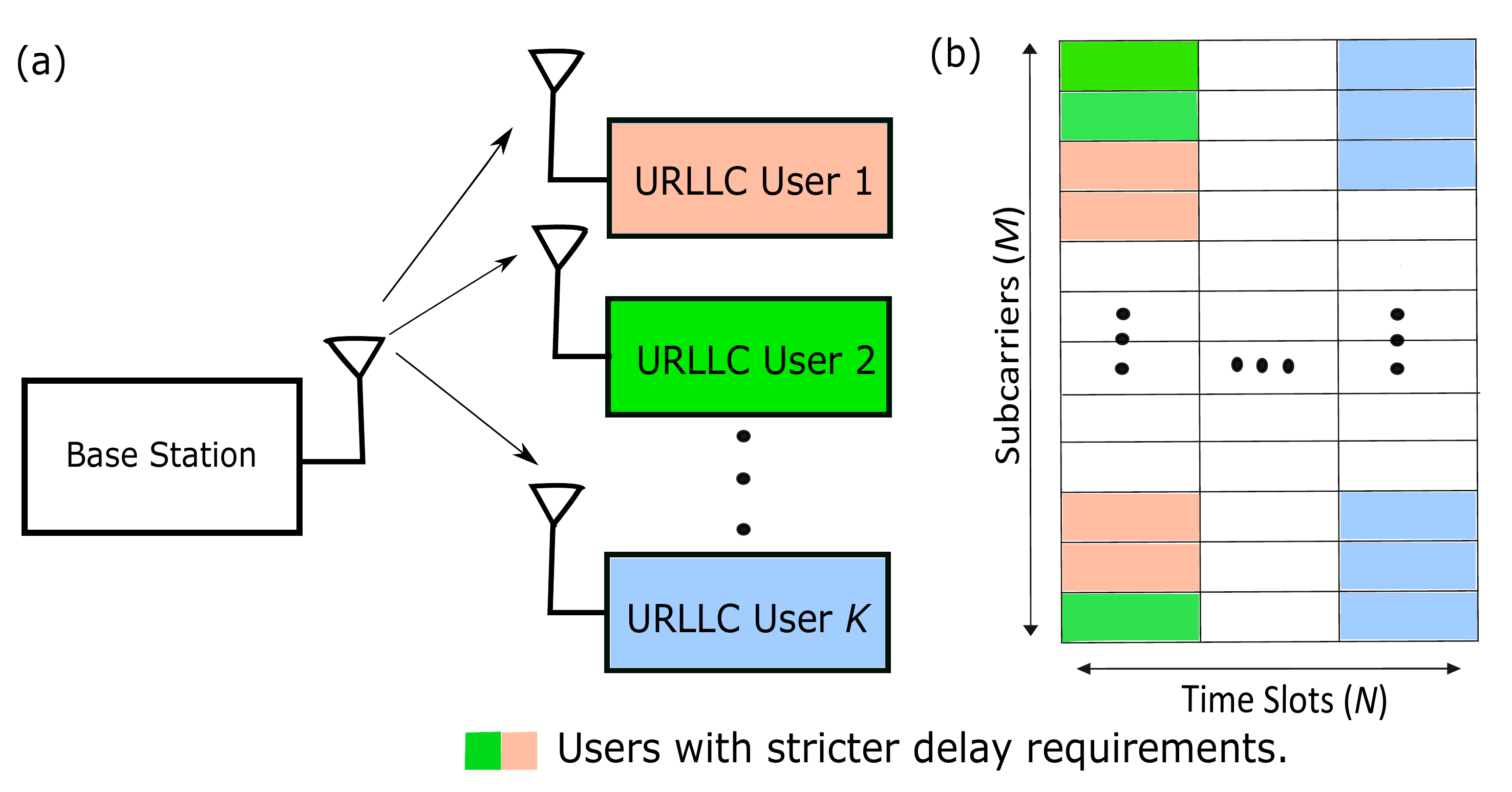}
  \centering
    \caption{ Multi-user downlink URLLC-OFDMA system: (a) System model with BS and $K$ users;    
      (b) Frame structure.}
    \label{model}
     \vspace{-3mm}
\end{figure}
\subsection{Channel Model}
 In this paper, we assume that the coherence time is larger than $T_{\mathrm{f}}$. Therefore, the channel gain of a given subcarrier remains constant for all $N$ time slots. The received signal at user $k$ on subcarrier $m$ in time slot $n$ is  modeled as \vspace*{-1mm}
\begin{IEEEeqnarray}{lll}
y^{(m)}_{k}[n]=h^{(m)}_{k}x^{(m)}_{k}[n]+z^{(m)}_{k}[n],
\end{IEEEeqnarray}
where $x^{(m)}_{k}[n]$ denotes the symbol transmitted by the BS for user $k$ on subcarrier $m$ in time slot $n$. Moreover, $z^{(m)}_{k}[n]\sim \mathcal{CN}(0,\sigma^{2})$ denotes complex Gaussian noise with zero mean and variance $\sigma^{2}$, and $h^{(m)}_{k}$ represents the complex channel coefficient between user $k$ and the BS on subcarrier $m$. Moreover, for future use, we define the signal-to-noise ratio (SNR) of user $k$ on subcarrier $m$ in time slot $n$ as \vspace*{-1mm}
\begin{IEEEeqnarray}{lll}
\gamma^{(m)}_{k}[n]=p^{(m)}_{k}[n]g^{(m)}_{k},
\end{IEEEeqnarray}
where $p^{(m)}_{k}[n]=\mathcal{E}\{|x^{(m)}_{k}[n]|^{2}\}$ is the power allocated to user $k$, and $g^{(m)}_{k}=\frac{|h^{(m)}_{k}|^{2}}{\sigma^{2}}$. Furthermore, to obtain a performance upper bound for  URLLC-OFDMA systems, perfect channel state information (CSI) is assumed to be available at the BS for resource allocation.   
\section{Resource Allocation Problem Formulation}
In this section, we discuss the achievable rate for SPC, the QoS of the users, and the adopted system performance metric for resource allocation design. Furthermore, we formulate the proposed resource allocation optimization problem for URLLC-OFDMA systems.
\subsection{Achievable Rate for SPC}
Shannon's capacity theorem considers the asymptotic case where the packet length approaches infinity and the decoding error probability goes to zero \cite{shannon}. Thus, it cannot be used for resource allocation design for URLLC systems, as URLLC systems have to employ short packets to achieve low latency. Furthermore, as a result of using short packets, decoding errors become unavoidable. For performance evaluation of SPC, the so-called normal approximation for finite blocklength codes was developed in \cite{thesis}. Mathematically, the maximum number of bits $B$ conveyed in a packet of $L$ symbols with error probability $\epsilon$ can be approximated as\vspace*{-2mm}\cite[Eq. (4.277)]{thesis},\cite[Fig. 1]{Erseghe1}
\begin{IEEEeqnarray}{lll}\label{normalapproximation}
B=\sum_{i=1}^{L}\log(1+\gamma_{i})-Q^{-1}(\epsilon)\sqrt{\sum_{i=1}^{L}V_{i}},
\end{IEEEeqnarray}
 where $V_{i}$ is the channel dispersion, which for the complex AWGN channel is obtained as \cite{thesis}\vspace*{-2mm}
\begin{IEEEeqnarray}{lll}\label{dispersion}
V_{i}=a^{2}\bigg(1-\frac{1}{(1+\gamma_{i})^2}\bigg).
\end{IEEEeqnarray} 
Here, $a=\log(\text{e})$, $Q^{-1}(\cdot)$ is the inverse of the Gaussian Q-function $Q(x)=\frac{1}{\sqrt{2\pi}}\int_{x}^{\infty}\text{exp}{\left(-\frac{t^{2}}{2}\right)}\text{d}t$, and $\gamma_{i}$ is the SNR of the $i^{th}$ received symbol. In this paper, the resource allocation algorithm design for downlink URLLC-OFDMA systems is based on (\ref{normalapproximation}). Each resource element carries one symbol, and by allocating several resource elements from the available $M \times N$ resource elements to a given user, the number of bits received by the user with packet error probability $\epsilon$ can be determined based on (\ref{normalapproximation}). 
\subsection{QoS and System Performance Metrics}
The QoS requirements for URLLC users comprise a minimum number of received bits, denoted by $B_{k}$, a target packet error probability denoted by $\epsilon_{k}$, and the number of time slots required for transmitting the user's packet, denoted by $D_{k}$. According to (\ref{normalapproximation}), the total number of bits transmitted over resources allocated to user $k$ can be written as:
\begin{IEEEeqnarray}{lll}\label{BT}
\Psi_{k}(\mathbf{p}_{k},\mathbf{s}_{k})=F_{k}(\mathbf{p}_{k},\mathbf{s}_{k})-V_{k}(\mathbf{p}_{k},\mathbf{s}_{k}),
\end{IEEEeqnarray}
where
\begin{eqnarray}
F_{k}(\mathbf{p}_{k},\mathbf{s}_{k})&=&\sum_{m=1}^{M}\sum_{n=1}^{N}s^{(m)}_{k}[n]\log(1+\gamma^{(m)}_{k}[n]), \\
V_{k}(\mathbf{p}_{k},\mathbf{s}_{k})&=&Q^{-1}(\epsilon_{k})\sqrt{\sum_{m=1}^M \sum_{n=1}^N s^{(m)}_{k}[n]V^{(m)}_{k}[n]},\\
V^{(m)}_{k}[n]&=&a^{2}\bigg(1-\frac{1}{(1+{\gamma^{(m)}_{k}[n]})^2}\bigg),
\end{eqnarray}
where $s^{(m)}_{k}[n]=\{0,1\}$ is a binary assignment indicator for resource element $[m,n]$. If subcarrier $m$ in time slot $n$ is assigned to user $k$, we have $s^{(m)}_{k}[n]=1$, otherwise $s^{(m)}_{k}[n]=~0$. Furthermore, we assume that each subcarrier is allocated to at most one user to avoid multiple access interference. $\mathbf{p}_{k}~{\in  \mathbb{R}^{MN \times 1}}$ and $\mathbf{s}_{k}~{\in  \mathbb{R}^{MN \times 1}}$ are the collections of optimization variables $p_{k}^{(m)}[n]$ $\forall m,n$ and $s_{k}^{(m)}[n]$ $\forall m,n$, for user $k$, respectively. Moreover, the delay requirement of user $k$ can be ensured by assigning all symbols of user $k$ to the first $D_{k}$ time slots. In other words, users requiring low latency are assigned resource elements at the beginning of the frame, cf. Fig.~\ref{model}(b).    
\begin{remark}
We note that a user can start decoding as soon as it has received all OFDMA symbols that contain its data, i.e., after $D_{k}$ time slots.  \end{remark}
The weighted sum throughput of the entire system is given by: \vspace*{-2mm}
\begin{IEEEeqnarray}{lll} \label{eq1}
U(\mathbf{p},\mathbf{s}) =\sum_{k=1}^{K}w_{k}\Psi_{k}(\mathbf{p}_{k},\mathbf{s}_{k}) =F(\mathbf{p},\mathbf{s})-V(\mathbf{p},\mathbf{s}), 
\end{IEEEeqnarray}
where
\begin{IEEEeqnarray}{lll}
\hspace*{-4mm} F(\mathbf{p},\mathbf{s})=\sum_{k=1}^{K}w_{k}F_{k}(\mathbf{p}_{k},\mathbf{s}_{k}), \hspace*{-2mm} \quad V(\mathbf{p},\mathbf{s})=\sum_{k=1}^{K}w_{k}V_{k}(\mathbf{p}_{k},\mathbf{s}_{k}).
\end{IEEEeqnarray}

 Here, $w_{k}$ are weights which can be used to control the fairness among the users. Moreover, $\mathbf{p} \in  \mathbb{R}^{KMN \times 1}$  and $\mathbf{s} \in  \{0,1\}^{KMN \times 1}$  are the collections of optimization variables $p_{k}^{(m)}[n]$, $\forall k,m,n$, and $s_{k}^{(m)}[n]$, $\forall k,m,n$, respectively. 
\subsection{Optimization Problem Formulation}
In the following, we formulate the resource allocation optimization problem for the maximization of the weighted sum throughput of the entire system while meeting the QoS requirements of each user regarding the received number of bits, the reliability, and the latency. In particular, the  power and subcarrier allocation policies are determined by solving the following optimization problem:
\begin{IEEEeqnarray}{lll}
\label{optimization1}
&& \hspace*{-0mm} \underset { {\mathbf {p}},\mathbf {s}}{ \mathop {\mathrm {maximize}}\nolimits } \,\,\  F(\mathbf{p},\mathbf{s})-V(\mathbf{p},\mathbf{s}) \\
\hspace*{-0mm}  \mbox {s.t.} 
&& \hspace*{-0mm} \nonumber \mbox {C1: }  F_{k}(\mathbf{p}_{k},\mathbf{s}_{k})-V_{k}(\mathbf{p}_{k},\mathbf{s}_{k}) \geq B_{k},  \forall k,  \\
&& \hspace*{-0mm} \nonumber  \mbox {C2: } p^{(m)}_{k}[n] \geq 0,  \forall k,m,n, \\  
&& \hspace*{-0mm} \nonumber  \mbox {C3: }  \sum_{k=1}^K \sum_{m=1}^{M}  \sum_{n=1}^{N}  s^{(m)}_{k}[n]p^{(m)}_{k}[n]  \leq  P_{\text{max}}, \\
&& \hspace*{-0mm} \nonumber \mbox {C4: } s^{(m)}_{k}[n] = \{0,1\}, \forall k,m,n, \quad \mbox {C5: } \sum_{k=1}^{K}s^{(m)}_{k}[n] \leq 1, \forall m,n, \\
&& \hspace*{-0mm} \nonumber \mbox {C6: } s^{(m)}_{k}[n]=0, \forall n>D_{k},  \forall k. 
\end{IEEEeqnarray}

 In (\ref{optimization1}), constraint $\mbox {C1}$ guarantees the transmission of a minimum number of $B_{k}$ bits to user $k$. Constraint $\mbox {C2}$ is the non-negative transmit power constraint. Constraint $\mbox {C3}$ is the total power budget constraint. Constraints $\mbox {C4}$ and $\mbox {C5}$ are imposed to ensure that each subcarrier in a given time slot is allocated to only one user. Finally, constraint $\mbox {C6}$ ensures that  user $k$ is served within $D_{k}$ time slots to meet its delay requirement. 
 
       The optimization problem in (\ref{optimization1}) is a  mixed integer non-convex optimization problem. The non-convexity is caused by the objective function, constraint $\mbox {C1}$, and the integer constraint for subcarrier allocation in $\mbox {C4}$. In general, mixed integer non-convex optimization problems are difficult to solve optimally in polynomial time. Hence, in the next section, we focus on developing a sub-optimal solution, where successive convex approximation is employed for computational efficiency.
\section{Resource Allocation Algorithm Design} 
In this section, we propose a low-complexity sub-optimal algorithm to solve problem (\ref{optimization1}). The proposed resource allocation algorithm design tackles the non-convexity of (\ref{optimization1}) in three main steps as outlined in the following.
\begin{itemize}
		\setlength{\itemsep}{1pt}
\item \textbf{Step 1\textendash Big-M formulation:} One reason for the non-convexity of (\ref{optimization1}) is the joint power and subcarrier allocation which introduces non-convex multiplicative terms of the form $s_{k}^{(m)}[n]p_{k}^{(m)}[n]$, e.g., constraint~$\mbox {C3}$. To address this issue, we employ the big-M formulation to relax these constraints by decomposing the multiplicative terms. 
\item \textbf{Step 2\textendash Integer relaxation:} We relax the integer constraint $\mbox {C4}$ by rewriting it in an equivalent form.  
\item \textbf{Step 3\textendash Difference of convex programming:} Difference of convex (DC) programming will be used to convexify the non-convex objective function and non-convex constraint $\mbox {C1}$. The resulting problem can be solved using standard convex optimization solvers, e.g., CVX \cite{cvx}, iteratively.  
\end{itemize}
 \textbf{Step 1:} To solve (\ref{optimization1}) efficiently, we introduce a new variable as follows:
\begin{IEEEeqnarray}{lll}\label{product}
\bar{p}_{k}^{(m)}[n]=s_{k}^{(m)}[n]p_{k}^{(m)}[n]. 
\end{IEEEeqnarray}
Therefore, the optimization problem in (\ref{optimization1}) can be rewritten in the following equivalent form:
\begin{IEEEeqnarray}{lll}\label{optimization2}
& \quad \underset { {\bar{\mathbf{p}}},\mathbf {s}}{ \mathop {\mathrm {maximize}}\nolimits }~ \bar{F}(\bar{\mathbf{p}})-\bar{V}(\bar{\mathbf{p}})\\
& \nonumber\mbox {s.t.}~~\mbox {C1:} \ \bar{F}_{k}(\bar{\mathbf{p}}_{k}) \hspace*{-0.5mm} - \hspace*{-0.5mm} \bar{V}_{k}(\bar{\mathbf{p}}_{k}) \hspace*{-0.5mm} \geq \hspace*{-0.5mm} B_{\textit{k}}, \forall k,  \quad \mbox {C2:} 
 \ \bar{p}^{(m)}_{k}[n] \hspace*{-0.5mm} \geq \hspace*{-0.5mm} 0,  \forall k,m,n, \\
&\qquad \nonumber \mbox {C3:} \ \sum_{k=1}^K \sum_{m=1}^{M}\sum_{n=1}^{N} \bar{p}^{(m)}_{k}[n]\leq P_{\text{max}}, \quad \mbox {C4},\mbox {C5},\mbox {C6}, 
\end{IEEEeqnarray}
 where
\begin{eqnarray}
\bar{F}(\bar{\mathbf{p}}) &=& \sum_{k=1}^{K}w_{k}\bar{F}_{k}(\bar{\mathbf{p}}_{k})\ , \quad \bar{V}(\bar{\mathbf{p}}) = \sum_{k=1}^{K}w_{k}\bar{V}_{k}(\bar{\mathbf{p}}_{k}), \\[-1mm]
\bar{F}_{k}(\bar{\mathbf{p}}_{k})&=&\sum_{m=1}^{M}\sum_{n=1}^{N}\log(1+\bar{p}^{(m)}_{k}[n]g_{k}^{(m)}),\\[-1mm]
\bar{V}_{k}(\bar{\mathbf{p}}_{k})&=&Q^{-1}(\epsilon_{k})\sqrt{\sum_{m=1}^M \sum_{n=1}^N \bar{V}^{(m)}_{k}[n]},\\[-1mm]
\bar{V}^{(m)}_{k}[n]&=&a^{2}\bigg(1-\frac{1}{(1+{\bar{p}^{(m)}_{k}[n]g_{k}^{(m)})}^2}\bigg).
\end{eqnarray} 
%
Here, $\bar{\mathbf{p}} \in  \mathbb{R}^{KMN \times 1}$ is the collection of the new variables ${\bar{p}}_{k}^{(m)}[n]$, $\forall k,m,n$. Since optimization problems (\ref{optimization1}) and (\ref{optimization2}) are equivalent, we will focus on solving (\ref{optimization2}). To design an efficient algorithm, we employ the big-M formulation\footnote{ For more details on the big M-formulation,
	 please refer to \cite[Section~2.3]{Leemixed}.} to decompose the product terms in (\ref{product}) \cite{yan}. Specifically, we impose the following additional constraints: 
\begin{IEEEeqnarray}{lll} &\hspace{-4mm} \mbox {C7}: \bar{p}^{(m)}_{k}[n]\leq P_{\text{max}} s^{(m)}_{k}[n], \ \ \  \forall k,m,n, \\\ & \hspace{-4mm} \mbox {C8}: \bar{p}^{(m)}_{k}[n]\leq p_{k}^{(m)}[n], \ \ \ \forall k,m,n,\\ & \hspace{-4mm} \mbox {C9}: \bar{p}_{k}^{(m)}[n]\geq p_{k}^{(m)}[n] -(1-s_{k}^{(m)}[n])P_{\text{max}}, \   \forall k,m,n,\quad \hspace{-4mm} \\ & \hspace{-4mm}\mbox {C10}: \bar{p}_{k}^{(m)}[n]\geq 0, \ \ \forall k,m,n.\end{IEEEeqnarray}
Note that constraints $\mbox {C7}-\mbox {C10}$ do not change the feasible set. 

 \textbf{Step 2:} The integer constraint $\mbox {C4}$ in (\ref{optimization2}) is a non-convex constraint. Thus, we rewrite constraint $\mbox {C4}$ in the following equivalent form: 
\begin{eqnarray}\hspace{-4mm}
&&\hspace*{-6mm} \mbox {C4a}: \sum_{k=1}^{K}\sum_{m=1}^{M}\sum_{n=1}^{N}s_{k}^{(m)}[n]-\sum_{k=1}^{K}\sum_{m=1}^{M}\sum_{n=1}^{N}({s_{k}^{(m)}}[n])^{2}\leq 0, \\
&&\hspace*{-6mm} \mbox {C4b}: 0 \leq s_{k}^{(m)}[n]\leq 1,   \ \ \  \forall m,n,k.
\end{eqnarray}
To facilitate the presentation, we rewrite constraint $\mbox {C4a}$ as follows
\begin{IEEEeqnarray}{lll}
\mbox {C4a}: W(\mathbf {s})-E(\mathbf {s})\leq 0,
\end{IEEEeqnarray}
where 
\begin{eqnarray}
W(\mathbf{s})\hspace*{-0.5mm}=\hspace*{-0.5mm}\sum_{k=1}^{K}\hspace*{-0.5mm}\sum_{m=1}^{M}\hspace*{-0.5mm}\sum_{n=1}^{N} \hspace*{-0.5mm} s_{k}^{(m)}[n], \quad E(\mathbf{s}) \hspace*{-0.5mm} = \hspace*{-0.5mm} \sum_{k=1}^{K}\hspace*{-0.5mm}\sum_{n=1}^{M}\hspace*{-0.5mm}\sum_{n=1}^{N}({s_{k}^{(m)}}[n])^{2}.
\end{eqnarray}
Now, the optimization variables $s_{k}^{(m)}[n]$ are continuous values between zero and one. Optimization problem (\ref{optimization2}) can be reformulated in the following equivalent form: 
\begin{IEEEeqnarray}{lll}\label{eeq17}&\underset { {\bar{\mathbf{p}}},\mathbf{p},\mathbf {s}}{ \mathop {\mathrm {minimize}}\nolimits }~ -\bar{F}(\bar{\mathbf{p}})+\bar{V}(\bar{\mathbf{p}}), \\ \nonumber &\mbox {s.t.}\qquad \mbox {C1}-\mbox {C3},\mbox {C4a},\mbox {C4b}, \mbox {C5}-\mbox {C10}. \end{IEEEeqnarray}
However, constraint $\mbox {C4a}$ is still a non-convex constraint because it is a difference of two convex functions\cite{kwan1,yan,Joinoptimization}. 
 In order to deal with this constraint, we introduce the following theorem.
\begin{thm}
Assuming a large value of $\beta\gg1$, the optimization problem in (\ref{eeq17}) is equivalent to the following optimization problem:
\begin{IEEEeqnarray}{lll}\label{eeq16}&\underset { {\bar{\mathbf{p}}},\mathbf{p},\mathbf {s}}{ \mathop {\mathrm {minimize}}\nolimits }~ -\bar{F}(\bar{\mathbf{p}})+\bar{V}(\bar{\mathbf{p}})+\beta\left(W(\mathbf {s})-E(\mathbf {s})\right), \\&\mbox {s.t.}~~ \nonumber\mbox {C1:}\ \bar{F}_{k}(\mathbf{\bar{p}}_{k})-\bar{V}_{k}(\mathbf{\bar{p}}_{k})\geq B_{k}, \  \forall \ k, \\&\qquad \nonumber \mbox {C2},\mbox {C3},\mbox {C4b},\mbox {C5}-\mbox {C10}. \end{IEEEeqnarray}
\end{thm}
\vspace{-3mm}
\begin{proof}
Please refer to Appendix~ A.
\end{proof}
 Constant $\beta$ in (\ref{eeq16}) is a large constant which is used as a penalty factor to penalize the objective function for any value of $s_{k}^{(m)}[n]$ that is not equal to $0$ or $1$. \\
    \textbf{Step 3:}  The optimization problem in (\ref{eeq16}) is still  non-convex because of the objective function and constraint $\mbox {C1}$. First, we analyze the convexity of the objective function in the following lemma.
\begin{lem}
The objective function is the difference of two convex functions.
\end{lem}
\begin{proof}
Please  refer to Appendix~ B.
\end{proof}
      The optimization problem in (\ref{eeq16}) belongs to the class of DC programming problems, since its objective function can be written as the difference of two convex functions and constraint $\mbox {C1}$ can be expressed as the difference of two concave functions. We can obtain a first order approximation for convex function $E(\mathbf{s})$ and concave function $\bar{V}_{k}(\mathbf{\bar{p}}_{k})$ via Taylor series as follows
\begin{eqnarray}
\label{inequalit2}
E(\mathbf{s}) &\geq & E(\mathbf{s}^{({j})})+ \nabla_{\mathbf{s}}E(\mathbf{s}^{({j})})^{T}(\mathbf{s}-\mathbf{s}^{({j})}),\,\,\ \text{and} \\
\label{inequalit3}
\bar{V}_{k}(\mathbf{\bar{p}}_{k}) &\leq & \bar{V}_{k}(\mathbf{\bar{p}}_{k}^{({j})})+ \nabla_{\mathbf{\bar{p}}_{k}}\bar{V}_{k}(\mathbf{\bar{p}}_{k}^{({j})})^{T}(\mathbf{\bar{p}}_{k}-\mathbf{\bar{p}}_{k}^{({j})}),
\end{eqnarray}
where $\mathbf{s}^{({j})}$ and $\mathbf{\bar{p}}_{k}^{({j})}$ are initial points, and
\begin{IEEEeqnarray}{lll}
 \nabla_{\mathbf{s}}E(\mathbf{s}^{({j})})^{T}(\mathbf{s}-\mathbf{s}^{({j})}) \nonumber \\ =\sum_{k=1}^{K}\sum_{m=1}^{M}\sum_{n=1}^{N}2s^{(m)({j})}_{k}[n]\left( s^{(m)}_{k}[n]-s^{(m)({j})}_{k}[n]\right),
 \end{IEEEeqnarray}
and
\begin{eqnarray}\nonumber \nabla_{\mathbf{\bar{p}}_{k}}\bar{V}_{k}(\mathbf{\bar{p}}_{k}^{({j})})
= \frac{a^{2} 
 Q^{-1}(\epsilon_{k})}{\sqrt{\sum_{m=1}^{M}\sum_{n=1}^{N} \bar{V}^{(m)}_{k}[n]}}\begin{pmatrix} 
  \ \frac{g_{k}^{(1)}}{(1+\bar{p}_{k}^{(1)}[1]g_{k}^{(1)})^{3}} \\
 \ \frac{g_{k}^{(2)}}{(1+\bar{p}_{k}^{(2)}[1]g_{k}^{(2)})^{3}} \\
  \vdots \\
  \frac{g_{k}^{(M)}}{(1+\bar{p}_{k}^{(M)}[N]g_{k}^{(M)})^{3}} \end{pmatrix}.
\end{eqnarray}
The right hand sides of (\ref{inequalit2}) and (\ref{inequalit3}) are affine functions. By substituting (\ref{inequalit2}) and (\ref{inequalit3}) into (\ref{eeq16}), we obtain the following convex optimization problem 
\begin{IEEEeqnarray}{lll} \label{eeeq16}
\underset { {\bar{\mathbf{p}}},\mathbf{p},\mathbf {s}}{ \mathop {\mathrm {minimize}}\nolimits }~ -\bar{F}(\bar{\mathbf{p}})+ \bar{V}(\bar{\mathbf{p}}^{({j})})+ \nabla_{\bar{\mathbf{p}}}\bar{V}(\bar{\mathbf{p}}^{({j})})^{T}(\bar{\mathbf{p}}-\bar{\mathbf{p}}^{({j})}) \quad \\  +\nonumber\beta\left(W(\mathbf{s})-E(\mathbf{s}^{({j})})- \nabla_{\mathbf{s}}E(\mathbf{s}^{({j})})^{T}(\mathbf{s}-\mathbf{s}^{({j})})\right), \\\mbox {s.t.}~~\nonumber\mbox {C1:} \ \bar{F}_{k}(\bar{\mathbf{p}})-\bar{V}_{k}(\bar{\mathbf{p}}^{({j})}) - \nonumber\nabla_{\bar{\mathbf{p}}}\bar{V}_{k}(\bar{\mathbf{p}}^{({j})})^{T}(\bar{\mathbf{p}}-\bar{\mathbf{p}}^{({j})})\geq B_{k}, \  \forall  k, \\\qquad \mbox {C2},\mbox {C3},\mbox {C4b}, \nonumber\mbox {C5}-\mbox {C10}. \end{IEEEeqnarray}
\par 
The optimization problem in (\ref{eeeq16}) is convex because the objective function is convex and the constraints span a convex set. Therefore, it can be efficiently solved by standard convex optimization solvers such as CVX \cite{cvx}. Algorithm 1 summarizes the main steps to solve (\ref{eeq16}) in an iterative manner, where the solution of (\ref{eeeq16}) in iteration $({j})$ is used as the initial point for the next iteration $({j}+1)$. The algorithm produces a sequence of improved feasible solutions until convergence to a local optimum point of problem (\ref{eeq16}) or equivalently problem (\ref{optimization1}) in polynomial time \cite{yan,Joinoptimization}.
\begin{algorithm}[t]
\caption{Successive Convex Approximation }
1:  \textbf{Initialize:} The maximum number of iterations $J_{\text{max}}$, iteration index $j=1$, penalty factor $\beta\gg1$, and  initial points $\mathbf{\bar{p}}^{(1)}$ and $\mathbf{{s}}^{(1)}$.\\
2: \textbf{Repeat}\\
3: Solve convex problem (\ref{eeeq16}) for a given $\mathbf{\bar{p}}^{({j})}$, $\mathbf{s}^{({j})}$\\
4: Set ${j}={j}+1$ and update $\mathbf{\bar{p}}^{({j}+1)}=\mathbf{\bar{p}}^{({j})}$, $\mathbf{{s}}^{({j}+1)}=\mathbf{{s}}^{({j})}.$ \\
5: \textbf{Until} ${j}=J_{\text{max}}$\\
6: \textbf{Return:} $\mathbf{\bar{p}}^{*}=\mathbf{\bar{p}}^{({J}_{\text{max}})}$, $\mathbf{s}^{*}=\mathbf{s}^{({J}_{{\text{max}}})}$.  
\label{sco2}
\end{algorithm}
\section{Performance Evaluation}
\ In this section, we provide simulation results to evaluate the effectiveness of the proposed resource allocation design for URLLC-OFDMA systems. The adopted simulation parameters are given in Table I, unless specified otherwise. In our simulations, the BS is located at the center of a cell. We consider the worst-case scenario where the users are located at the edge of the cell. For simplicity the user weights are set to $w_{k}=1,  \forall k=1,\dots,K$. The path loss is calculated as  $35.3 + 37.6 \log_{10}(d_{k})$\cite{chsecross}, where $d_{k}$ is the distance from the base station to user $k$. The subcarrier gains follow a Rayleigh distribution. The penalty factor in (\ref{eeeq16}) is set to $\beta=10 \log(1+\frac{P_{\text{max}}}{\sigma^{2}})$. All simulation results are averaged over $100$ realizations of the mutlipath fading. Moreover, the proposed algorithm converges to a local optimum point after a few iterations. We use ${J}_{{\text{max}}}=5$ iterations for our simulations.
\subsection{Performance Metric}
To evaluate the performance of the system, we define the sum throughput of the system for a given channel realization as follows:
\begin{IEEEeqnarray}{lll} \label{average}
\bar{R} = \begin{cases}\frac{1}{MN}\sum_{k=1}^{K} \Psi_{k}(\mathbf{p}_{k},\mathbf{s}_{k}),  \quad  &\text{if} \,\,   (\mathbf{p},\mathbf{s}) \,\, \text{is feasbile} \\ 0  \qquad & \text{otherwise.} \end{cases}
\end{IEEEeqnarray}
 If the optimization problem is infeasible for a given channel realization, we set the corresponding sum throughput to zero. The average system sum throughput is obtained by averaging $\bar{R}$ over all considered channel realizations.   
\subsection{Performance Bound and Benchmark Schemes}
\par  We compare the performance of the proposed resource allocation algorithm design with an upper bound and two benchmark schemes:
\begin{itemize}
	\setlength{\itemsep}{1pt}
 \item {\textbf{Upper bound}}: In this scheme, Shannon's capacity is used for optimization in (\ref{optimization1}), i.e., $V(\mathbf{p},\mathbf{s})$ and $V_{k}(\mathbf{p}_{k},\mathbf{s}_{k})$ are set to zero in the objective function and constraint $\mbox {C1}$, respectively, but all other constraints are retained. Successive convex optimization is used to solve the resulting new optimization problem. This scheme provides an (unachievable) upper bound for the average system sum throughput of the network.
     \item {\textbf{Benchmark scheme 1}}: In this case, the solution obtained from the upper bound is applied in (\ref{average}) which uses the normal approximation in (\ref{BT}) to compute the average sum throughput, i.e., Shannon's capacity is used for resource allocation design but the normal approximation is used for performance evaluation.     
  \item {\textbf{Benchmark scheme 2}}: In this scheme, we fix the transmit power of every subcarrier as $p_{k}^{(m)}[n]=\frac{P_{\text{max}}}{MN}$, i.e., equal power allocation is used. Therefore, the problem in (\ref{optimization1}) is reduced to a subcarrier assignment problem which is solved using successive convex optimization.
\end{itemize}

\begin{table}[t]
\centering
\caption{System parameters.} \vspace*{-2mm}
\label{tab:table}
\renewcommand{\arraystretch}{1.4}
\scalebox{0.7}{%
\begin{tabular}{|c||c|} 
        \hline
        Cell radius & 250 meters \\ \hline
        Number and bandwidth of subcarriers & 64 and 15 kHz \\ \hline
        Noise power density  & -174 dBm/Hz \\ \hline
Number of bits per packet  & 160 bits \\ \hline  
 Maximum base station transmit power $P_{\text{max}}$  &  $45$~dBm \\ \hline  
Number of iterations ${J}_{\text{max}}$ for Algorithm 1 &  5  \\ \hline  
\end{tabular}}
\vspace*{-5mm}
\end{table} 
\vspace*{-1mm}
  \subsection{Simulation Results}
   In Fig.~\ref{changepower}, we show the average sum throughput versus the maximum transmit power at the BS, $P_{\text{max}}$. The maximum packet error probability of all users is set to $\epsilon_{k}=10^{-7} \, \forall k$. There are $K=4$ users in the system, and the number of time slots is $N=6$. We assume that $D_{1}=2$, $D_{2}=3$, and $D_{3}=D_{4}=6$, i.e., the first and the second users require smaller delays than the remaining users. As can be observed, the average system sum throughput improves with the maximum transmit power $P_{\text{max}}$ because the SNR of all users is increased. Below a certain value of $P_{\text{max}}$, for each of the considered schemes, the average sum throughput is very small (almost zero). This is due to the strict QoS requirements of the URLLC users which cannot be met for small $P_{\text{max}}$, i.e., the resource allocation policy is infeasible. Furthermore, the value of $P_{\text{max}}$, below which the average sum throughput is almost zero, is $25$~dBm for the proposed scheme, whereas it is $27$~dBm and $30$~dBm for benchmark schemes 2 and 1, respectively. The upper bound requires less power for feasibility than the other schemes, because, for the upper bound, we set  $V_{k}(\mathbf{p}_{k},\mathbf{s}_{k}), \forall k$ to zero in $\mbox{C1}$. Benchmark schemes 1 and 2 require more power for feasibility than the proposed scheme, as the equal power allocation used for benchmark scheme 2 is not optimal, while for benchmark scheme 1, the resource allocation policies $ {\mathbf {p}}$ and $\mathbf {s}$ are determined based on Shannon's capacity formula which may violate constraint $\mbox{C1}$ in (\ref{optimization1}), especially for small $P_{\text{max}}$. Therefore, Shannon's capacity cannot be used for the design of URLLC-OFDMA systems, especially for low-to-medium $P_{\text{max}}$, since the QoS requirements cannot be guaranteed. For high $P_{\text{max}}$, for the proposed scheme, all non-zero $p_{k}^{(m)}[n]$ assume large values. Hence, the corresponding $\gamma^{(m)}_{k}[n]$ in (6)-(8) are large and $V_{k}(\mathbf{p}_{k},\mathbf{s}_{k})$ becomes negligible compared to $F_{k}(\mathbf{p}_{k},\mathbf{s}_{k})$. Therefore, in this case, benchmark scheme 1, which assumes $V_{k}(\mathbf{p}_{k},\mathbf{s}_{k})$ is zero, yields a similar performance as the proposed scheme. 
\begin{figure}[t]
\includegraphics[scale=0.5]{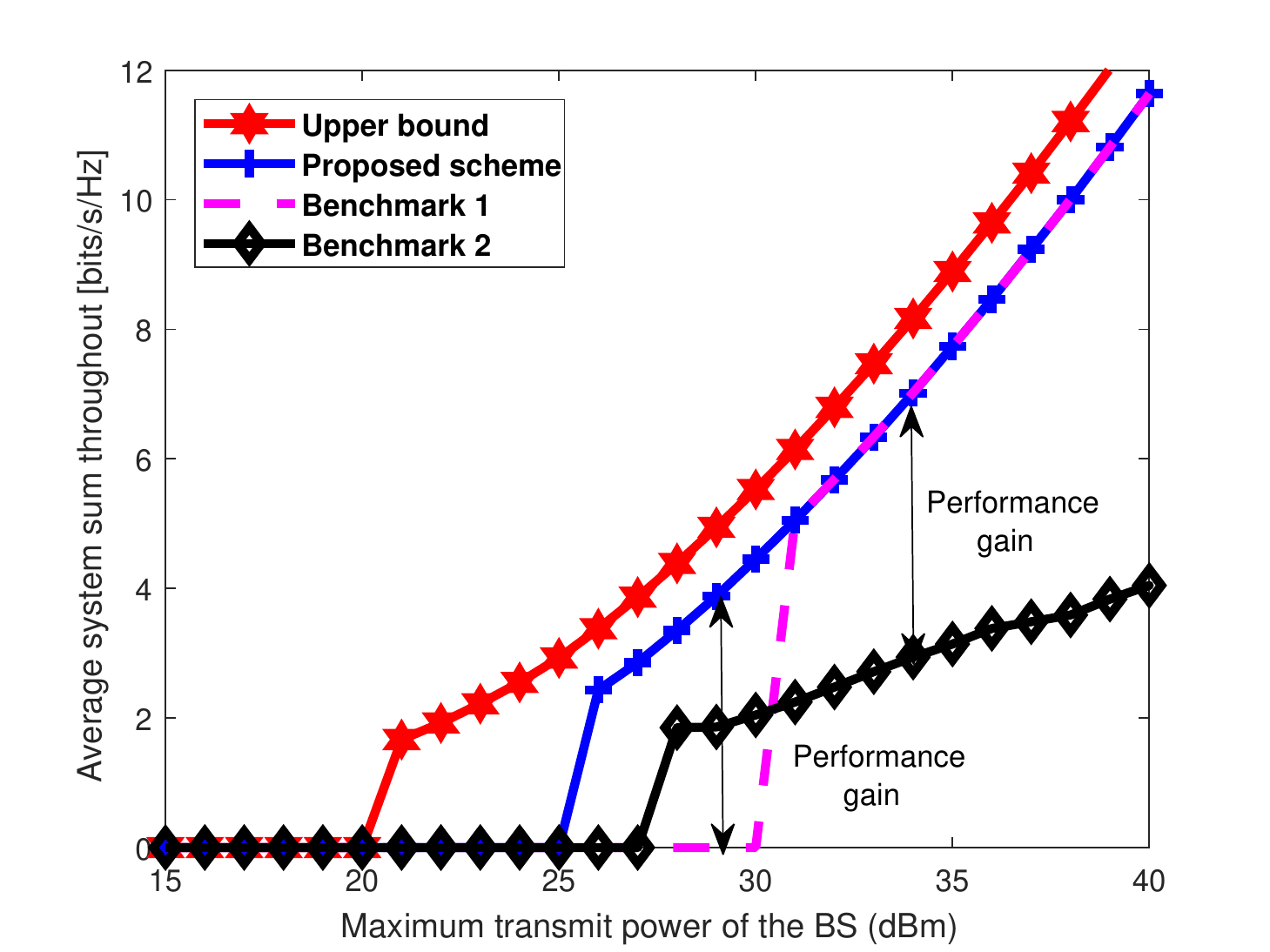}
 \centering
    \caption{ Average system sum throughput (bits/s/Hz) vs. maximum transmit power at the BS (dBm), $P_{\text{max}}$, for different resource allocation schemes and $K=4$ users.}
    \vspace{-6mm}
    \label{changepower}
\end{figure}
\begin{figure}[t]
\includegraphics[scale=0.5]{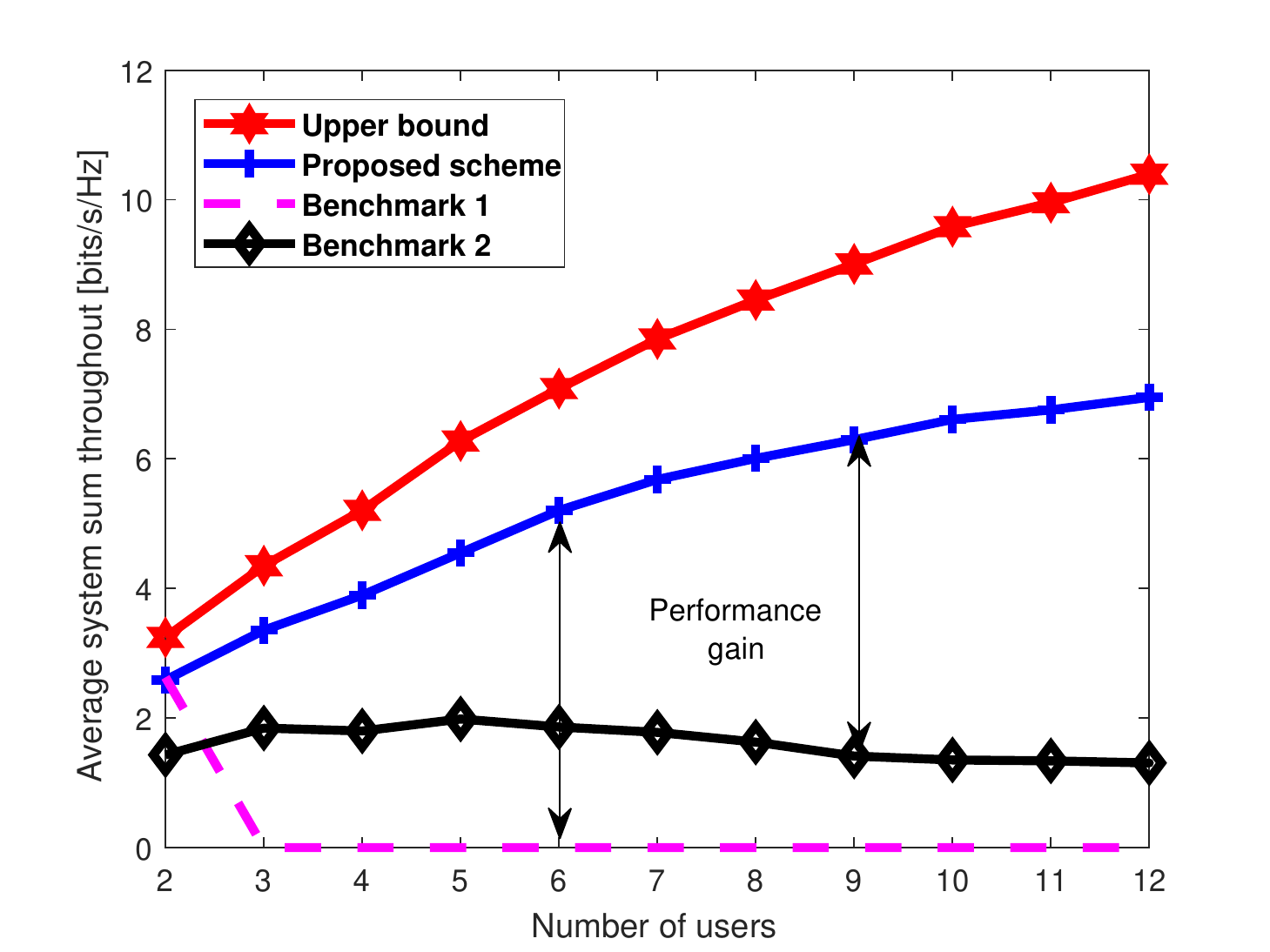}
 \centering
 \caption{Average system sum throughput (bits/s/Hz) vs. number of users for different resource allocation schemes and $P_{\text{max}}=~27$~dBm.}
  \vspace{-7mm}
    \label{users}
\end{figure}
 Fig.~\ref{users} shows the average sum  throughput versus the number of users for $P_{\text{max}}=27$~dBm. We assume that $N=4$ and the maximum packet error probability for all users is set to $\epsilon_{k}=10^{-6} \, \forall k$. The URLLC users' delay requirements are $D_{1}=2$ and $D_{k}=4$, $\forall$ $k\neq1$. As can be observed from Fig.~\ref{users}, for the considered parameters, for the upper bound and the proposed scheme, the average sum throughput increases with the number of users, as these schemes can exploit multi-user diversity. 
  However, benchmark scheme 1 fails to support even $K=3$ users due to the invalid design criterion (Shannon's capacity formula), which does not take into account the effect of short packet transmission on the achievable rate and reliability. As a result, the solution based on this design may violate constraint $\mbox{C1}$ in (\ref{optimization1}). For benchmark scheme 2, the average sum throughput increases up to $K=5$ users, but decreases if more users are added, because this scheme does not exploit all available degrees of freedom for resource allocation. 
  \section{Conclusion}
In this paper, we studied the resource allocation algorithm design for broadband URLLC-OFDMA systems. The resource allocation algorithm design was formulated as a non-convex optimization problem for maximization of the weighted system sum throughput subject to QoS constraints for the URLLC users. To achieve a favorable tradeoff between complexity and performance, a low complexity sub-optimal algorithm was developed to solve the optimization problem. Simulation results revealed that the proposed system design can support URLLC, and the proposed algorithm outperformed two benchmark schemes.\begin{appendices}
	\section{}
In the following, we show that problems (\ref{eeq17}) and (\ref{eeq16}) are equivalent. Let $U^{*}$ denote the optimal objective value of (\ref{eeq17}). We define the Lagrangian function, denoted by ${\mathcal {L}}(\bar{\mathbf {p}},\mathbf {s},\beta)$, as \cite{Boyed}
\begin{IEEEeqnarray}{lll}  \label{e25} \ {\mathcal {L}}(\bar{\mathbf {p}},\mathbf {s},\beta)= -\bar{F}(\bar {\mathbf {p}})+\bar{V}(\bar {\mathbf {p}}) +\beta(W(\mathbf{s})-E(\mathbf{s})),
\end{IEEEeqnarray}
where $\beta$ is the Lagrange multiplier corresponding to constraint $\mbox {C4a}$.  
Note that $W(\mathbf{s})-E(\mathbf{s}) \geq 0$ holds.
Using Lagrange duality \cite{Boyed}, we have the following relation \footnote{Weak duality holds for convex and non-convex optimization problems\cite{Boyed}.}
\begin{IEEEeqnarray}{lll} \label{eq28} &\hspace {-2pc}U_{d}^{*}=\underset {\beta \ge 0}{ \mathop {\mathrm {max}}\nolimits } \quad \underset {\bar {\mathbf {p}},\mathbf {p},\mathbf {s} \in \boldsymbol{\Omega }}{ \mathop {\mathrm {min}}\nolimits } \quad {\mathcal {L}}(\bar {\mathbf {p}},\mathbf {s},\beta) \IEEEyesnumber \IEEEyessubnumber
\\\overset {(a)}{\le }&\underset {\bar {\mathbf {p}},\mathbf {p},\mathbf {s} \in \boldsymbol{\Omega }}{ \mathop {\mathrm {min}}\nolimits } \quad \underset {\beta \ge 0}{ \mathop {\mathrm {max}}\nolimits } \quad {\mathcal {L}}(\bar {\mathbf {p}},\mathbf {s},\beta) = U^{*},\IEEEyessubnumber
\end{IEEEeqnarray}
where $\boldsymbol{\Omega }$ is the feasible set specified by the constraints in (\ref{eeq17}).  
In the following, we first prove the strong duality, i.e., $U_{d}^{*}=U^{*}$. 
Let $(\bar{\mathbf{p}}^{*},\mathbf{p}^{*},\mathbf{s}^{*}, \beta^{*} )$ denotes the solution of (34a). For this solution, the following two cases are possible. \textit{Case 1)} If $W(\mathbf{s})-E(\mathbf{s})>0$ holds, the optimal $\beta^{*}$ is infinite. Hence, $U_{d}^{*}$ is infinite too, which contradicts the fact that it is upper bounded by a finite-value $U^{*}$.  \textit{Case 2)} If $W(\mathbf{s})-E(\mathbf{s})=0$ holds, then $(\bar{\mathbf{p}}^{*},\mathbf{p}^{*},\mathbf{s}^{*})$ belongs to the feasible set of the original problem (\ref{eeq16}) which implies $U_{d}^{*}=U^{*}$. Hence, strong duality holds, and we can focus on solving the dual problem (34a) instead of the primal problem (34b). 

Next, we show that any $\beta \geq \beta_{0}$ is an optimal solution for dual problem (34a), i.e., $\beta^{*}$, where $\beta_{0}$ is some sufficiently large number. To do so, we show that ${\Theta}(\beta)\triangleq \underset {\bar {\mathbf {p}}, \mathbf {p},\mathbf {s} \in \boldsymbol{\Omega }}{ \mathop {\mathrm {min}}\nolimits} \quad {\mathcal {L}}(\bar {\mathbf {p}},\mathbf {s},\beta)$ is a monotonically increasing function of $\beta$. Recall that $W(\mathbf{s})-E(\mathbf{s})\geq 0$ holds for any given $(\bar{\mathbf{p}},\mathbf {s}) \in \boldsymbol{\Omega }$. Therefore, ${\mathcal {L}}(\bar {\mathbf {p}},\mathbf {s},\beta_{1}) \leq {\mathcal {L}}(\bar {\mathbf {p}},\mathbf {s},\beta_{2})$ holds for any given $(\bar{\mathbf{p}},\mathbf {s}) \in \boldsymbol{\Omega }$ and $0 \leq \beta_{1}\leq \beta_{2}$. This implies ${\Theta}(\beta_{1})\leq {\Theta}(\beta_{2})$   and that ${\Theta}(\beta)$ is monotonically increasing in $\beta$. Using this result, we can conclude that  
$\Theta (\beta) =U^{*},\, \forall \beta \ge  {\beta }_{0}$.

In summary, due to strong duality, we can use the dual problem  (\ref{eeq16}) to find the solution of the primal problem (\ref{eeq17}) and any $\beta \geq \beta_{0}$ is an optimal dual variable. These results are concisely given in Theorem 1 which concludes the proof.
\section{}
In the following, we show that the objective function of (\ref{eeq16}) is the difference of two convex functions. To this end, we rewrite it as follows \vspace*{-2mm}
\begin{IEEEeqnarray}{lll}\bar{U}(\bar{\mathbf{p}},\mathbf {s})=-\bar{F}(\bar{\mathbf{p}})+\bar{V}(\bar{\mathbf{p}})+\beta\left(W(\mathbf {s})-E(\mathbf {s})\right). \end{IEEEeqnarray}
 $\bar{U}(\bar{\mathbf{p}},\mathbf {s})$ is the difference of two convex functions if $\bar{F}(\bar{\mathbf{p}})$ is concave, $\bar{V}(\bar{\mathbf{p}})$ is concave, and both $W(\mathbf {s})$ and $E(\mathbf {s})$ are convex. Function $\bar{F}(\mathbf{\bar{p}})$ is a sum of logarithmic functions, and hence, it is a concave function\cite{Boyed}. Moreover, $W(\mathbf {s})-E(\mathbf {s})$ is a difference of two convex functions \cite{yan}. Furthermore, to prove  $\bar{V}(\mathbf{\bar{p}})$ is concave, we rewrite it as follows \vspace*{-2mm}
\begin{IEEEeqnarray}{lll}
 \bar{V}(\mathbf{\bar{p}})=\sum_{k=1}^{K}w_{k}Q^{-1}(\epsilon_{k})\sqrt{\sum_{m=1}^{M}\sum_{n=1}^{N}\bar{V}^{(m)}_{k}[n]} ,
 \end{IEEEeqnarray}  
where \vspace*{-2mm}
\begin{IEEEeqnarray}{lll}
   \bar{V}^{(m)}_{k}[n]= a^{2}\left(1-\left(\frac{1}{1+\bar{p}^{(m)}_{k}[n]g_{k}^{(m)}}\right)^2\right).
\end{IEEEeqnarray} 
 Note that $\bar{V}(\mathbf{\bar{p}})$ is always positive, because for $\epsilon \in (0,0.5)$, $Q^{-1}(\epsilon_{k}) >0$ holds. 
  To prove the concavity of $\bar{V}(\mathbf{\bar{p}})$, first we will show that  $\bar{V}^{(m)}_{k}[n]$ is concave by taking the first and second derivatives with respect to $\bar{p}^{(m)}_{k}[n]$  as follows:  
\begin{eqnarray}
\frac{\mathrm{d}\bar{V}^{(m)}_{k}[n]}{\mathrm{d}\bar{p}^{(m)}_{k}[n]}&=&\frac{ 2a^{2}g_{k}^{(m)}}{(1+\bar{p}^{(m)}_{k}[n]g_{k}^{(m)})^3}, \\
\frac{\mathrm{d}^2\bar{V}^{(m)}_{k}[n]}{{\mathrm{d}(\bar{p}^{(m)}_{k}[n]})^2}&=&\frac{-6 a^{2}(g_{k}^{(m)})^2}{(1+\bar{p}^{(m)}_{k}[n]g_{k}^{(m)})^4}.
\end{eqnarray} 
 Function $\bar{V}^{(m)}_{k}[n]$ is concave because the second derivative is negative for any $\bar{p}^{(m)}_{k}[n]>0$. Moreover, since a sum of concave functions is also concave, $\sum_{m=1}^{M}\sum_{n=1}^{N}\bar{V}^{(m)}_{k}[n]$ is concave. By using the composition rules of convex analysis, the square root is concave and the extended-value extension on the real line is non-decreasing \cite{Boyed}. Thus, the square root of a concave function is concave. Finally, a weighted sum of concave functions is also concave. This concludes the proof. 
\end{appendices}
\bibliography{ref}  
\bibliographystyle{IEEEtran}
\end{document}